\documentclass[journal]{IEEEtran}
\usepackage{ifpdf}
\usepackage{cite}

\ifCLASSINFOpdf
   \usepackage[pdftex]{graphicx}
\else
\fi
\usepackage[cmex10]{amsmath}
\usepackage{geometry}
\usepackage{chngpage}
\usepackage{graphicx}
\usepackage{amsmath}
\usepackage{amsthm}
\usepackage{amsfonts}
\usepackage{amssymb}
\usepackage{latexsym}
\usepackage[ruled,vlined]{algorithm2e}

\usepackage{array}

\newtheorem{teo}{Theorem}%[subsection]
\newtheorem{cor}[teo]{Corollary}
\newtheorem{lem}[teo]{Lemma}

\newtheorem{defn}[teo]{Definition}
\newtheorem{Exemp}[teo]{Example}

\begin{document}

\title{On Fractional Decoding of Reed-Solomon Codes}

\author{{Welington~Santos

\thanks{W. Santos is with the Programa de Pós-Graduação em Matemática, Universidade Federal do Paraná, Caixa Postal 19081, 81531-990, Curitiba-PR Brazil. His study was financed in part by the Coordenação de Aperfeiçoamento de Pessoal de Nível Superior- Brasil (CAPES)-Finance Code 001. Email: wsantos.math@gmail.com}}}% <-this % stops a space
\maketitle

\begin{abstract}
We define a virtual projection of a Reed-Solomon code $RS(q^{l},n,k)$ to an $RS(q,n,k)$ Reed-Solomon code. A new probabilistic decoding algorithm that can be used to perform fractional decoding beyond the $\alpha$- decoding radius is considered. %Fractional decoding means that we want to recover the original data from the corrupted codeword in the case that the decoder can download only an $\alpha$-proportion of the codeword. 
An upper bound for the failure probability of the new algorithm is given, and the performance is illustrated by examples.

%We consider a probabilistic approach for fractional decoding of Reed-Solomon codes beyond the $\alpha$-decoding radius.  %Finally, we create an Interleaved Reed-Solomon code and use collaborative decoding to increase the $\alpha$-decoding radius.
\end{abstract}


\IEEEpeerreviewmaketitle

\section{Introduction}
  \IEEEPARstart{A}{n} Interleaved Reed-Solomon code \cite{D.Bleichenbacher,Schmidt,Krachkovsky} is obtained by stacking $m$ codewords of different $m$ $ RS(q,n,k_{j})$ codes of the same length $n$. A codeword of an  Interleaved Reed-Solomon code is an $m\times n$ matrix over the field $\mathbb{F}_{q}$. Interleaved Reed-Solomon codes make sense in scenarios where the error affects all $m$ $ RS$ codewords at the same positions. In \cite{SchmidtBeyond}, Schmidt et al. presented a scheme that virtually extends a low-rate $ RS$ code to an Interleaved Reed-Solomon code and a probabilistic decoding algorithm that can correct errors beyond the decoding radius of the usual $RS$-code.
  
  Recently, Tamo at al. \cite{BarM}, considered error correction by maximum distance separable (MDS) codes based on part of the received codeword to define a fractional decoding problem, and the $\alpha$-decoding radius of an $(n,k,l)$ array code over a finite field $\mathbb{F}_{q}$. The fractional decoding problem is motivated by the fact that in distributed systens \cite{Dimakis}, usually there is a limitation on the disk operation as well as on the amount of information transmitted for the purpose of decoding.
   % defined fractional decoding problem, and $\alpha$-decoding radius of an $(n,k,l)$ array code over $\mathbb{F}_{q}$. %In \cite{BarM} it was also shown that Reed-Solomon codes $ RS(q^{l},n,k)$ with evaluation set $\mathcal{L}\subseteq\mathbb{F}_{q}$ achives the $\alpha$-decoding radius.
  
  In this constribuition, we consider a Reed-Solomon code $RS(q^{l},n,k)$ with evaluation set $\mathcal{L}\subseteq\mathbb{F}_{q}$ and define a virtual projection to an $ RS(q,n,k)$ Reed-Solomon code. We also present a probabilistic approach to the problem of fractional decoding. For  $\alpha =m/l$ and an $ RS(q^{l},n,k)$ code of rate $R\leqslant\frac{\alpha}{m(1-\alpha)+1}$ our method corrects more errors than guaranteed by the $\alpha$-decoding radius with failure probability given approximately by $mq^{-(m+1)(\tau_{P_{\alpha}}-t)-1}$, where $t<\tau_{P_{\alpha}}$ is the number o errors that we would like to correct.
  
  This work is structured as follows. In Sect. 2, we recall collaborative decoding of Interleaved Reed-Solomon codes \cite{Schmidt} and fractional decoding \cite{BarM}. In Sect. 3, we define a virtual projection to an $ RS(q,n,k)$-code, and we show how the virtual projection can be used to perform fractional decoding beyond the $\alpha$-decoding radius.
\section{Preliminaries}
\subsection{Reed-Solomon and Interleaved Reed-Solomon Codes}
\begin{defn}
	Let $\mathcal{L}=\lbrace\alpha_{1},\ldots,\alpha_{n}\rbrace$ where $ \alpha_{1},\ldots,\alpha_{n}$ are distinct nonzero elements of the finite field $\mathbb{F}_{q}$. For a given univariante polynomial $f(x)\in\mathbb{F}_{q}[x]$ denote
	\[f(\mathcal{L})=\left(f(\alpha_{1}),\ldots,f(\alpha_{n})\right).\]
 A Reed-Solomon code $ RS(q,n,k)$ over a field $\mathbb{F}_{q}$ with $n<q$ is given by
	\begin{equation}
	 RS(q,n,k)=\lbrace c=f(\mathcal{L}):\ f(x)\in\mathbb{F}_{q}[x]_{k}\rbrace,
	\end{equation}
	where $\mathbb{F}_{q}[x]_{k}$ denotes the set of all univeriante polynomials of degree less than $k$. The set $\mathcal{L}$ is called the evaluation set of $\mathcal{C}$.
\end{defn}

An Interleaved Reed-Solomon code of order $m$ is given by $m$ underlying $ RS$ codes, which are arranged in a matrix form.

\begin{defn}
	Let the set $\mathcal{K}=\left\lbrace k_{0},k_{2},\ldots,k_{m-1}\right\rbrace$, consist of $m$ integers, where all $k_{j}<n$. An \textit{Interleaved Reed-Solomon code} $ IRS(q,n,\mathcal{K},m)$ of order $m$ is given by
	\begin{equation}
	 IRS(q,n,\mathcal{K},m)=\left\lbrace C=\left(\begin{array}{l}
	f_{0}(\mathcal{L})\\
	f_{1}(\mathcal{L})\\
	\ \ \ \vdots\\
	f_{m-1}(\mathcal{L})
	\end{array}\right):\ f_{j}(x)\in\mathbb{F}_{q}[x]_{k_{j}}
	\right\rbrace,\end{equation}
	The codewords $f_{j}(\mathcal{L})\in  RS(q,n,k_{i})$ are called elementary codewords of the $ IRS(q,n,\mathcal{K},m)$-code. If the dimensions $k_{j}$ are equal  $\text{for all}\ j=0,\ldots, m-1$ the $ IRS$ code is called \textit{homogneous}. Otherwise, the $ IRS$ code is called \textit{heterogeneous}.
\end{defn}
In considering $IRS$ codes we are interested in column errors. This is equivalent to transmission of the $ IRS$ code over a $q^{m}$-aray channel. 

Let $C\in IRS(q,n,\mathcal{K},m)$ and $R=C+E$, where $E=(E_{1},\ldots,E_{n})$ and $w(E):=|\lbrace i:E_{i}\neq 0\rbrace|=t$, denote the received word. The $m$ elementary codewords of an $ IRS$ code are affected by $m$ elementary error words $e^{(0)},e^{(1)},\ldots,e^{(m-1)}$ of weight $wt(e^{(j)})=t_{j}\leqslant t$. Let $\mathcal{E}^{(j)}$ denote the set of error positions for the $j-th$ elementary received word. Since we are considering column erros, the union of the $m$ sets of error positions $\mathcal{E}=\mathcal{E}^{(0)}\cup\mathcal{E}^{(1)}\cup\ldots\cup\mathcal{E}^{(m-1)}\subseteq\left\lbrace 1,\ldots,n\right\rbrace$ has cardinality $|\mathcal{E}|=t$.

\subsection{Collaborative Decoding of Interleaved Reed-Solomon Codes}

In \cite{Schmidt}, Schmidt et al. introduced the concept of collaborative decoding for Interleaved Reed-Solomon codes. This decoder is based on the fact that the errors occur in the same positions of each elementary codeword of the Interleaved Reed-Solomon code.

%Assume that the codewords of an $ IRS$ code are transmitted over a $q^{m}$-ary channel.
In the first step of collaborative decoding, $m$ syndrome polynomials $S^{(0)}(x),S^{(1)}(x),\ldots,S^{(m-1)}(x)$ of degree smaller than $n-k_{j}$ are calculated. The syndrome polynomial is
\begin{equation}
S^{(j)}(x)=\sum_{i=1}^{n-k_{j}}S_{i}^{(j)}x^{i-1}
\end{equation}
with coefficients:
\begin{equation}
S_{i}^{(j)}=r^{(j)}(\alpha_{i}^{k_{j}})=\sum_{h=1}^{n}r_{h}^{(j)}\alpha_{i}^{k_{j}(h-1)}
\end{equation}
for all $i=1,\ldots,n-k_{j}$ and $j=0,\ldots,m-1$.

%Collaborative decoding calculates the \textit{error locator polynomial} $\Lambda(x)\in\mathbb{F}_{q}[x]$ with $\Lambda(\alpha_{i}^{-1})=0$ for all $i\in\mathcal{E}$. This polynomial can be normalized, e.g, such that it is monic: $\Lambda(x)=\Lambda_{1}+\Lambda_{2}x+\ldots+\Lambda_{t}x^{t-1}+x^{t}$. Since $m$ error vectors $e^{(0)},\ldots,e^{(m-1)}$ are non-zero at the same positions, the error locator polynomial $\Lambda(x)$ is the same for all $m$ received words and it can be found by solving the following common key equation:
%\begin{equation}
%S^{(i)}(x)\Lambda(x)\equiv\Omega^{(i)}(x)\mod x^{n-k_{i}},\ i=0,\ldots,m-1.
%\end{equation}
As in the classical case, these syndromes are used to form a linear system of equations $S\Lambda=T$, 
\begin{equation}\label{sitme}
\left(\begin{array}{c}
S^{(0)}\\
S^{(1)}\\
\vdots\\
S^{(m-1)}
\end{array}\right).\left(\begin{array}{c}
\Lambda_{1}\\
\Lambda_{2}\\
\vdots\\
\Lambda_{t}
\end{array}\right)=\left(\begin{array}{c}
T^{(0)}\\
T^{(1)}\\
\vdots\\
T^{(m-1)}
\end{array}\right),
\end{equation}
where each sub-matrix $S^{(j)}$ is a $(n-k_{j}-t)\times t$ matrix and each $T^{(j)}$ is a column vector of length $n-k_{j}-t$:

\begin{equation*}
S^{(j)}=\left(\begin{array}{cccc}
S^{(j)}_{t+1}&S^{(j)}_{t}&\cdots&S^{(j)}_{2}\\
S^{(j)}_{t+2}&S^{(j)}_{t+1}&\cdots&S^{(j)}_{3}\\
\vdots&\vdots&\ddots&\vdots\\
S^{(j)}_{n-k_{j}}&S^{(j)}_{n-k_{j}-1}&\cdots&S^{(j)}_{n-k_{j}-t+1}
\end{array}\right),
\end{equation*}
\begin{equation}
T^{(j)}=\left(\begin{array}{c}
-S^{(j)}_{1}\\
-S^{(j)}_{2}\\
\vdots\\
-S^{(j)}_{n-k_{j}-t}
\end{array}\right)
\end{equation}

The system of equations (\ref{sitme}) has $\sum_{j=0}^{m-1}(n-k_{j}-t)$ equations and $t$ unknowns. In order to guarantee unambiguous decoding, the number of linearly independent equations has to be greater than or equal to the number of unknowns. Under the assumption that all equations in (\ref{sitme}) are linearly independent we obtain the following restriction on $t$:
\begin{equation}\label{eq9}
\sum_{j=0}^{m-1}(n-k_{j}-t)\geqslant t
\end{equation}

Which can be rewritten as

\begin{equation}\label{jointerror}
t\leqslant\frac{m}{m+1}\left(n-\frac{1}{m}\sum_{i=0}^{m-1}k_{j}\right):=\tau_{ IRS}
\end{equation}

However, there is a certain probability that some of the equations (\ref{sitme}) are linearly dependent. In this case, there is no unique solution of the system of equations and we declare a \textit{decoding failure}.

A collaborative decoder for $IRS$ codes corrects $t\leqslant\tau_{IRS}$ errors with probability at least (\cite{Schmidt}, Theorem 7).

\begin{equation}
1-\left(\frac{q^{m}-\frac{1}{q}}{q^{m}-1}\right)^{t}\frac{q^{-(m+1)(\tau_{IRS}-t)}}{q-1}.
\end{equation}

\subsection{Fractional decoding}

Tamo at al. \cite{BarM}, introduced the concept of fractional decoding where error correction by maximum distance separable codes based on part of the received codeword is considered. The idea is that the decoder downloads an $\alpha$ proportion of each of the codeword's coordinates. Below we will describe the $\alpha$-decoding problem.

Fractional decoding is defined in the following
%An $(n,k,l)$ array code $\mathcal{C}$ is formed of $l\times n$ matrices $C=\left(C_{1},\ldots,C_{n}\right)\in\left(\mathbb{F}_{q}^{l}\right)^{n}$, where $\mathbb{F}_{q}$ is a finite field. Each column $C_{i}$ of the matrix is a codeword coordinate, and the parameter $l$ that determines the dimension of the column vector $C_{i}$ is called \textit{sub-packetization}. Array codes may also be consider as an $(n,k)$-code over the alphabet $\mathbb{F}_{q}^{l}$, and then one error amounts to an incorrect column $C_{i}$. Correcting up to $t$ errors means correcting any combination of errors $E=\left(E_{1},\ldots,E_{n}\right)\in\left(\mathbb{F}_{q}^{l}\right)^{n}$ of Hamming weight $w(E):=|\left\lbrace i:E_{i}\neq 0\right\rbrace |\leqslant t$, where the received codeword is the matrix $C+E=(C_{i}+E_{i}, i=1,\dots,n)$.

\begin{defn}
	Let $\mathcal{C}$ be an $(n,k,l)$ array code over field $\mathbb{F}_{q}$. We say that $\mathcal{C}$ corrects up to $t$ errors by downloading $\alpha nl$ symbols of $\mathbb{F}_{q}$ if there exist functions 
	\begin{equation}
	f_{i}:\mathbb{F}_{q}^{l}\longrightarrow\mathbb{F}_{q}^{\alpha_{i}l}, i=1,\ldots,n\ \text{and}\  g:\mathbb{F}_{q}^{\left(\sum_{i=1}^{n}\alpha_{i}\right)}\longrightarrow\mathbb{F}_{q}^{nl}
	\end{equation}
	such that $\sum_{i=1}^{n}\alpha_{i}\leqslant n\alpha$ and for any codeword $C\in\mathcal{C}$ and any error $E\in\left(\mathbb{F}_{q}^{l}\right)^{n}, w(E)\leqslant t$
	\begin{equation}
	g(f_{1}(C_{1}+E_{1}),\ldots,f_{n}(C_{n}+E_{n}))=(C_{1},\ldots,C_{n}).
	\end{equation}
	For $\alpha\geqslant k/n$, we define the $\alpha$-decoding radius of $\mathcal{C}$ as the maximum number of errors that $\mathcal{C}$ can correct by downloading $\alpha nl$ symbols of $\mathbb{F}_{q}$, and denote it as $r_{\alpha}(\mathcal{C})$.
	
	Define the $\alpha$-decoding radius $r_{\alpha}(n,k)$ as follows:
	\begin{equation}
	r_{\alpha}(n,k)=\max\lbrace r_{\alpha}(\mathcal{C}):\mathcal{C}\ \text{is an (n,k)-code}\rbrace .
	\end{equation}
	
\end{defn}

Given an $(n,k)$-linear code we should take $\alpha\geqslant\frac{k}{n}$ because the codeword encodes $k$ data symbols, and even without errors to recover the data the decoder needs at least as many imput symbols. If $\alpha
=1$, we return to the standard problem, so the goal of fractional decoding is study error correction for $\alpha$ in the range $\frac{k}{n}\leqslant\alpha < 1$.

It was also shown in \cite{BarM} that the $\alpha$-decoding radius of a $(n,k)$-linear code is 
\begin{equation}\label{alpharadius}
\tau_{\alpha}=\left\lfloor\frac{n-k/\alpha}{2}\right\rfloor
\end{equation}
and that an $ RS(q^{l},n,k,\mathcal{L})$ with $\mathcal{L}\subseteq\mathbb{F}_{q}$ achieves the optimal $\alpha$-decoding radius (\ref{alpharadius}). 

\section{Fractional Decoding and collaborative decoding}

%In this section we present the concept of virtual projection of a Reed-Solomon code over $\mathbb{F}_{q^{l}}$ with evaluation set $\mathcal{L}\subseteq\mathbb{F}_{q}$ to a Reed-Solomon code over $\mathbb{F}_{q}$ and show how it can be used to do fractional decoding beyond the $\alpha$-decoding radius.

\subsection{Virtual Projection to an Irterleaved Reed-Solomon Code}
Schmidt et al. \cite{SchmidtBeyond,Schmidt_Vladimir}, suggested to extend a low-rate $RS(n,k)$ code to an $ IRS$ code to perform syndrome decoding of the $ RS(n,k)$ code beyond half the minimum distance, of course, with some failure probability. Zeh et al. \cite{Zeh}, defined the mixed virtual extension of a \text{homogeneous interleaved Reed-Solom code} to an  \text{heterogeneous interleaved Reed-Solom code}  with objective of decoding beyond half its joint error-correcting capability \cite{D.Bleichenbacher}.

In this subsection, we will introduce the concept of virtual projection of a Reed-Solomon code $RS(q^{l},n,k)\subseteq\mathbb{F}_{q^{l}}^{n}$ with evaluation set $\mathcal{L}=\left\lbrace\alpha_{1},\ldots,\alpha_{n}\right\rbrace\subseteq\mathbb{F}_{q}$ to a heterogeneous Reed-Solomon code $ IRS(q,n,\mathcal{K},m)$. Our purpose is to use the virtual projection to perform fractional decoding beyond the $\alpha$-decoding radius.
\begin{defn}
	Let $A_{0},A_{1},\ldots,A_{m-1}\subseteq\mathbb{F}_{q}$ be $m$ pairwise disjoint sets of the field $\mathbb{F}_{q}$. For $j=0,1,\ldots, m-1$, define the annihilator polynomials of the set $A_{j}$ to be
	\begin{equation}
	p_{j}(x)=\prod_{\omega\in A_{j}}\left(x-\omega\right).
	\end{equation}
	Note that, $\deg p_{j}(x)=|A_{j}|\ \forall j=0,\ldots,m-1$.
\end{defn}

\begin{defn}
	Let $F=\mathbb{F}_{q^{l}}$ be a finite field extension of $B=\mathbb{F}_{q}$ of degree $l$. The field trace is defined
	\[ tr_{F/B}(\beta)=\beta+\beta^{q}+\beta^{q^{2}}+\ldots+\beta^{q^{l-1}}.\]
	Let $\zeta_{0},\zeta_{1},\ldots,\zeta_{l-1}$ be a basis of $F$ over $B$, and let $\nu_{0},\nu_{1},\ldots,\nu_{l-1}$ be the dual basis, then
	\[\beta=\sum_{i=0}^{l-1}tr_{F/B}(\zeta_{i}\beta)\nu_{i}.\]
	In other words, any element $\beta$ in $F$ can be calculated from its $l$ projections $\left\lbrace tr_{F/B}(\zeta_{i}\beta)\right\rbrace_{i=0}^{l-1}$ on $B$.
\end{defn}

\begin{defn}
	Given a polynomial $h(x)=a_{k-1}x^{k-1}+a_{k-2}x^{k-2}+\ldots+a_{0}\in\mathbb{F}_{q^{l}}[x]$ and $A_{0},\ldots,A_{m-1}\subseteq\mathbb{F}_{q}$ $m$ pairwise disjoint subsets of $\mathbb{F}_{q}$. Define
	
	\begin{IEEEeqnarray}{rCl}\label{1}
		T_{j}(h)(x)& = & h_{l-m+j}(x)(p_{j}(x))^{(l-m)(j+1)}
		\nonumber\\
		&& \qquad + \sum_{u=0}^{l-m-1}h_{u}(x)(p_{j}(x))^{u(j+1)}
		%	\label{eq:sizecorr1}
	\end{IEEEeqnarray}
	for all $ j=0,\ldots,m-1$ and the polynomial $h_{i}(x)\in\mathbb{F}_{q}[x]$ is given by
	\begin{equation}
	h_{i}(x)=tr(\zeta_{i}a_{k-1})x^{k-1}+tr(\zeta_{i}a_{k-2})x^{k-2}+\ldots+tr(\zeta_{i}a_{0}).
	\end{equation}
\end{defn}

\begin{lem}\label{lematransf}
	Let $\mathcal{C}=RS(q^{l},n,k)$ be a Reed-Solomon code and $h(\mathcal{L})\in\mathcal{C}$ where $\mathcal{L}\subset\mathbb{F}_{q}$ is the evaluation set of $\mathcal{C}$. Then each codeword $T_{j}(h)(\mathcal{L})$ is a codeword of the Reed-Solomon code
	\begin{equation}
	\mathcal{C}_{j}=RS\left(q,n,k+|A_{j}|(l-m)(j+1)\right).
	\end{equation}
\end{lem}
\begin{proof}
	Fist note that
	\begin{IEEEeqnarray*}{rCl}
		\deg T_{j}(h)(x)& \leqslant & \max\left\lbrace\deg h_{l-m+j}(x)(p_{j}(x))^{(l-m)(j+1)},\right.
		\nonumber\\
		&&\left. \deg\sum_{u=0}^{l-m-1}h_{u}(x)(p_j(x))^{u(j+1)}\right\rbrace
		%	\label{eq:sizecorr1}
	\end{IEEEeqnarray*}
	and we can check that
	
	\begin{IEEEeqnarray*}{rCl}
		\deg h_{l-m+j}(x)(p_{j}(x))^{(l-m)(j+1)}&=& \deg h_{l-m+j}(x)
		\nonumber\\
		&& +|A_{j}|(l-m)(j+1)\\
		&<& k+|A_{j}|(l-m)(j+1).
		%	\label{eq:sizecorr1}
	\end{IEEEeqnarray*}
	and
	\begin{eqnarray*}
		\deg \sum_{u=0}^{l-m-1}h_{u}(x)(p_{j}(x))^{u(j+1)}%&\leqslant&\deg h_{m-l-1}(x)(p_{j}(x))^{(l-m-1)(j+1)}\\
		%&=&k+|A_{j}|(l-m-1)(j+1)-1\\
		&<&k+|A_{j}|(l-m)(j+1).
	\end{eqnarray*}
	So, $\deg T_{j}(h)(x)< k+|A_{j}|(l-m)(j+1)\  \text{for all}\  j=0,1,\ldots m-1.$
	Now we must check that $T_{j}(h)(\mathcal{L})\in\mathbb{F}_{q}^{n}$. By definition, $T_{j}(h)(\mathcal{L})=\left(T_{j}(h)(\alpha_{1}),\ldots,T_{j}(h)(\alpha_{n})\right)$, so we just need to prove that $T_{j}(h)(\alpha_{i})\in\mathbb{F}_{q}\ \text{for all}\ i=1,\ldots,n$. For all $j=0,\ldots,m-1.$ we have
	
	\begin{IEEEeqnarray*}{rCl}
		T_{j}(h)(\alpha_{i})&=&h_{m-l+j}(\alpha_{i})(p_{j}(\alpha_{i}))^{(l-m)(j+1)}
		\nonumber\\
		&& +\sum_{u=0}^{l-m-1}h_{u}(\alpha_{i})(p_{j}(\alpha_{i}))^{u(j+1)}
		%	\label{eq:sizecorr1}
	\end{IEEEeqnarray*}
	as $h_{u}(x),p_{j}(x)\in\mathbb{F}_{q}[x]\ \text{and}\ \alpha_{i}\in\mathbb{F}_{q}$ it is clear that $T_{j}(h)(\alpha_{i})\in\mathbb{F}_{q}\ \text{for all}\ i=1,\ldots,n\ \text{and}\  j=0,\ldots,m-1$.
\end{proof}

\begin{defn}\label{projection}
	Let $\mathcal{C}=RS(q^{l},n,k)$ be a Reed-Solomon code with evaluation set $\mathcal{L}=\left\lbrace \alpha_{1},\ldots,\alpha_{n}\right\rbrace\subseteq\mathbb{F}_{q}$ and let $A_{0},\ldots,A_{m-1}$ any subsets of $\mathbb{F}_{q}$ such that $\sum_{j=0}^{m-1}|A_{j}|\geqslant k$. The Virtual Projection $\mathcal{C}_{P_{m/l}}(q,n,\mathcal{K})$ is given by
	\begin{equation}\label{VPIRS}
	\mathcal{C}_{P_{m/l}}=\left\lbrace\left(\begin{array}{c}
	c_{(0)}\\
	c_{(1)}\\
	\vdots\\
	c_{(m-1)}
	\end{array}\right)=\left(\begin{array}{c}
	T_{0}(h)(\mathcal{L})\\
	T_{1}(h)(\mathcal{L})\\
	\vdots\\
	T_{m-1}(h)(\mathcal{L})
	\end{array}\right)\right\rbrace,
	\end{equation}
where $T_{j}(h)(x)$ is given by (\ref{1}) and $\mathcal{K}=\left\lbrace k_{0},\ldots,k_{m-1}\right\rbrace$ with $k_{j}=k+|A_{j}|(l-m)(j+1)\ \text{for all}\ j=0,\ldots,m-1.$
\end{defn}
Assume that a codeword $c(\mathcal{L})\in\mathcal{C}$ is transmitted over a noisy channel, which adds $t$ erros in such a way, that the word $y(\mathcal{L})=c(\mathcal{L})+e(\mathcal{L})$ is observed at the channel output. Using the observed word $y(\mathcal{L})$, we calculate the $m$ polynomials $T_{j}(y)(x)$, $j=0,\ldots,m-1$, and create the matrix
\begin{equation}\label{Y}
Y=\left(\begin{array}{ccc}
T_{0}(y)(\alpha_{1})&\ldots&T_{0}(y)(\alpha_{n})\\
T_{1}(y)(\alpha_{1})&\ldots&T_{1}(y)(\alpha_{n})\\
\vdots&\ddots&\vdots\\
T_{m-1}(y)(\alpha_{1})&\ldots&T_{m-1}(y)(\alpha_{n})
\end{array}\right)
\end{equation}

The matrix $Y$ can be considered as received word of the virtual projection $\mathcal{C}_{P_{m/l}}(q,n,\mathcal{K})$ of $\mathcal{C}=RS(q^{l},n,k)$.

\begin{teo}
	Let $c(\mathcal{L})\in RS(q^{l},n,k)$ be a codeword of a Reed-Solomon code $\mathcal{C}$ transmitted over a noisy channel. Assume that the word $y(\mathcal{L})=c(\mathcal{L})+e(\mathcal{L})$ is received, if $e=(e_{1},\ldots,e_{n})$ has $t$ nonzero coefficients $e_{i_1},\ldots,e_{i_{t}}$ then the matrix $Y$ is a corrupted codeword of the $\mathcal{C}_{P_{m/l}}(q,n,\mathcal{K})$ code with at most $t$ erroneous columns at the positions $i_{1},\ldots,i_{t}$.
\end{teo}
\begin{proof}
	If $e=0$, then $y=c\in\mathcal{C}$, and by Lemma (\ref{lematransf}) we know that $Y$ is a codeword of the virtually projection $\mathcal{C}_{P_{m/l}}(q,n,\mathcal{K})$. Note that
	\begin{equation*}
	T_{j}(y)(\alpha_{i})=T_{j}(c+e)(\alpha_{i})=T_{j}(c)(\alpha_{i})+T_{j}(e)(\alpha_{i}).
	\end{equation*}
	Clearly, if $e_{i}=0$, that is, if $i\notin\left\lbrace i_{1},\ldots,i_{t}\right\rbrace$, then $T_{j}(e)(\alpha_{i})=0$ for all $j=0,\ldots,m-1$. If $i\in\left\lbrace i_{1},\ldots,i_{t}\right\rbrace$, then $T_{j}(e)(\alpha_{i})$ may be non-zero, so $Y$ has at most $t$ erroneous columns.
\end{proof}
Unlike the virtual extension to an $ IRS$ code \cite{Schmidt_Vladimir}, where it is possible to ensure that given a word $y=c+e$ the virtual extension of $y$ is a word with exactly $t$ erroneous columns, in the virtual projection we can not assure it.

In addition, in the virtual extension approach given a codeword $c\in RS(q,n,k)$ and its virtual extension $C\in IRS$ when we recover the word $C\in IRS$, we immediately recover the codeword $c\in  RS(q,k)$ (the first row of the codeword $C$). In virtual projection it is not so immediately that given a codeword $c\in  RS(q^{l},n,k)$ and its virtual projection $C\in\mathcal{C}_{P_{m/l}}$ we can recover the codeword $c\in RS(q^{l},n,k)$ just by recovering the codeword $C\in\mathcal{C}_{P_{m/l}}$, but the following ensures it.

\begin{lem}\label{Lemarecu}
	Given polynomials $\left\lbrace T_{j}(h)(x)\right\rbrace_{j=0}^{m-1}$ as in (\ref{1}). Suppose that $\sum_{j=0}^{m-1}|A_{j}|\geqslant\deg h(x)$ then we can recover the polynomials $\left\lbrace h_{j}(x)\right\rbrace$ and consequently we can recover $h(x)$.
\end{lem}
\begin{proof}
	$T_{j}(h)(\omega)=h_{0}(\omega)$ for all $\omega\in A_{j}$; of course, we can rewrite (\ref{1}) as 
	
	\begin{IEEEeqnarray*}{rCl}
		T_{j}(h)(x)& = & h_{l-m+j}(x)(p_{j}(x))^{(l-m)(j+1)}
		\nonumber\\
		&&+\sum_{u=0}^{l-m-1}h_{u}(x)(p_{j}(x))^{u(j+1)}\\
		&=&h_{l-m+j}(x)(p_{j}(x))^{(l-m)(j+1)}
		\nonumber\\
		&&+h_{0}(x)(p_{j}(x))^{0(j+1)}+\sum_{u=1}^{l-m-1}h_{u}(x)(p_{j}(x))^{u(j+1)}.
		%	\label{eq:sizecorr1}
	\end{IEEEeqnarray*}
	So, $T_{j}(h)(\omega)=h_{0}(\omega)\ \text{for all}\ \omega\in A_{j}$. Then, we know the evaluations of $h_{0}(\omega)$ at all the points $\cup_{j=0}^{m-1}A_{j}$ and by assumption, $\sum_{j=0}^{m-1}|A_{j}|\geqslant\deg h(x)\geqslant\deg h_{0}(x)$, so we can recover $h_{0}(x)$. Now from $h_{0}(x)$ and $\left\lbrace T_{j}(h)(x)\right\rbrace_{j=0}^{m-1}$, we can calculate the polynomials
	
	\begin{IEEEeqnarray*}{rCl}
		T_{j}^{(1)}(h)(x)& = & \frac{T_{j}(h)(x)-h_{0}(x)}{p_{j}(x)^{j+1}}\\
	%	&=&h_{l-m+j}(x)(p_{j}(x))^{(l-m-1)(j+1)}
	%	\nonumber\\
	%	&&+\sum_{u=1}^{l-m-1}h_{u}(x)(p_{j}(x))^{(u-1)(j+1)}\\
		&=& h_{l-m+j}(x)(p_{j}(x))^{(l-m-1)(j+1)}
		\nonumber\\
		&& +h_{1}(x)+\sum_{u=2}^{l-m-1}h_{u}(x)(p_{j}(x))^{(u-1)(j+1)}.
		%	\label{eq:sizecorr1}
	\end{IEEEeqnarray*}
	
	So, $T_{j}^{(1)}(h)(\omega)=h_{1}(\omega)\ \text{for all}\ \omega\in A_{j}$, and again, we know the evaluation of $h_{1}(x)$ in $\cup_{j=0}^{m-1}A_{j}$. So, we can recover $h_{1}(x)$. From $h_{0}(x),h_{1}(x)\ \text{and}\ \left\lbrace T_{j}(h)(x)\right\rbrace_{j=0}^{m-1}$ we can calculate the polynomials
	\[T_{j}^{(2)}=\frac{T_{j}^{(1)}(h)(x)-h_{1}(x)}{p_{j}(x)^{j+1}}.\]
	Since $T_{1}^{(2)}(h)(\omega)=h_{2}(\omega)$ for all $\omega\in A_{j}$, by the previous argument we can recover $h_{2}(x)$. Generally, the polynomials $\left\lbrace h_{l-m+j}(x)\right\rbrace_{j=0}^{m-1}$ can be recovered from
	\[h_{l-m+j}(x)=\frac{T_{j}(h)(x)-\sum_{u=0}^{l-m-1}h_{u}(x)(p_{j}(x))^{u(j+1)}}{(p_{j}(x))^{(l-m)(j+1)}}.\]
\end{proof}

By Lemma \ref{Lemarecu}, we conclude that given an $ RS(q^{l},n,k)$-code with evaluation set $\mathcal{L}\subseteq\mathbb{F}_{q}$ and its virtual projection $\mathcal{C}_{P_{m/l}}$ it is possible to recover a codeword $c\in\mathcal{C}$ using the code $\mathcal{C}_{P_{m/l}}$ whenever the received word $y=c+e$ has no more than $t$ errors with $t<\tau_{P_{m/l}}$, where $\tau_{P_{m/l}}$ denotes the decoding radius of $\mathcal{C}_{P_{m/l}}$. Hence, we have the following algorithm.
\begin{algorithm}
	\caption{Virtual Projection IRS Decoder}
	\textbf{Input:} Received word $y(\mathcal{L})=c(\mathcal{L})+e(\mathcal{L})$, $\alpha=m/l$
	
	\textbf{For:} $j=0\ \text{to}\ m-1$ \textbf{do}
	
	Create the matrix $Y$ from $T_{j}(y)(\mathcal{L})$ and calculate the syndromes $S^{(0)},\ldots,S^{(m-1)}$. 
	
	Compute $t$ and $\Lambda(x)$ by Algorithm 1 in  \cite{Schmidt} .
	
	\eIf{$t<\tau_{P_{\alpha}}$ and $\Lambda(x)$ is $t$-valid}{\textbf{for} each $j$ from $0$ to $m-1$ \textbf{do}
		
		evaluate errors, and calculate $T_{j}(e)(\mathcal{L})$
		
		calculate $T_{j}(\hat{c})(\mathcal{L})=T_{j}(y)(\mathcal{L})-T_{j}(e)(\mathcal{L})$
		
		Use Lemma \ref{Lemarecu} to compute $c(\mathcal{L})$
	}{decoding failure}
	
	\textbf{output:} $c(\mathcal{L})\in\mathcal{C}$ or decoding failure
\end{algorithm}

\begin{teo}\label{VPIRSR}
	Let $\mathcal{C}= RS(q^{l},n,k)$ be a Reed-Solomon code then its virtual projection code $\mathcal{C}_{P_{m/l}}(q,n,\mathcal{K})$ given by Definition \ref{projection} has maximum decoding radius $\tau_{P_{m/l}}$ given by
	\begin{equation}\label{Projectionradius}
	\tau_{P_{m/l}}=\frac{m}{m+1}\left(n-k-\frac{(l-m)}{m}\sum_{j=0}^{m-1}|A_{j}|(j+1)\right).
	\end{equation}
\end{teo}
\begin{proof}
	The decoding radius of the $\mathcal{C}_{P_{m/l}}(q,n,\mathcal{K})$ code is the error-correcting radius of the heterogeneous $IRS(q,n,\mathcal{K},m)$ code where $\mathcal{K}=\left\lbrace k_{0},\ldots,k_{m-1}\right\rbrace$ and $k_{j}=k+|A_{j}|(l-m)(j+1)\ \text{for all}\ j=0,\ldots,m-1$. The correcting radius is given by (\ref{jointerror})
	
	\begin{eqnarray*}
		\tau_{P_{m/l}}&=&\frac{m}{m+1}\left(n-\frac{1}{m}\sum_{j=0}^{m-1}k_{i}\right)\\
		%&=&\frac{m}{m+1}\left(n-\frac{1}{m}\sum_{j=0}^{m-1}(k+|A_{j}|(l-m)(j+1))\right)\\
		%&=&\frac{m}{m+1}\left(n-\frac{1}{m}\left(mk+(l-m)\sum_{j=0}^{m-1}|A_{j}|(j+1)\right)\right)\\
		&=&\frac{m}{m+1}\left(n-k-\frac{l-m}{m}\sum_{j=0}^{m-1}|A_{j}|(j+1)\right).
	\end{eqnarray*}
\end{proof}

\begin{cor}\label{corola1}
	Let $\mathcal{C}= RS(q^{l},n,k)$ be a Reed-Solomon code and $\mathcal{C}_{P_{m/l}}(q,n,\mathcal{K})$ its virtual projection as in (\ref{VPIRS}), then:
	\begin{itemize}
		\item[i)] If $l=m$ then $\tau_{P_{m/l}}=\frac{l}{l+1}(n-k)$;
		\item[ii)] If $l=m=1$ then $\tau_{P_{m/l}}=\frac{n-k}{2}=\tau$;
		\item[iii)] If $|A_{j}|=b$ for all $j=0,\ldots,m-1$ then $$\tau_{P_{m/l}}=\frac{m}{m+1}\left(n-k-b\frac{(l-m)}{m}\binom{m+1}{2}\right).$$
	\end{itemize}
\end{cor}
\begin{proof}
	Straight forward calculation from (\ref{Projectionradius}).
\end{proof}
Note that if, $l=m$ then $\tau_{P_{m/l}}$ is the decoding radius of a \textit{homogeneous} Interleaved Reed-Solomon code  \cite{Schmidt, Schmidt_Vladimir}. For $l=m=1$ the result $\tau_{P_{m/l}}$ is the decoding radius of the $ RS(q,n,k)$ Reed-Solomon code over $\mathbb{F}_{q}$.

\subsection{Fractional decoding beyond the $\alpha$-decoding radius}
Let $\mathcal{C}= RS(q^{l},n,k)$ a Reed-Solomon code with evaluation set $\mathcal{L}=\left\lbrace\alpha_{1},\ldots,\alpha_{n}\right\rbrace\subseteq\mathbb{F}_{q}.$ Let $\alpha=m/l$, where $m$ and $l$ are positive integers and $m|k$. We will show that is possible to perform fractional decoding beyond the $\alpha$-decoding radius. 

Let $c=(c_{1},\ldots,c_{n})=(h(\alpha_{1}),\ldots,h(\alpha_{n}))\in\mathcal{C}$, where $h(x)\in\mathbb{F}_{q^{l}}[x]_{k}$. Let also $A_{0},\ldots,A_{m-1}\subseteq\mathbb{F}_{q}$ be $m$ pairwise disjoint subsets of $\mathbb{F}_{q}$, each of size $k/m$. The $m$ symbols we download from the $i$-th coordinate are

\begin{IEEEeqnarray}{rCl}\label{download}
	d_{i}^{j}& = & tr_{\mathbb{F}_{q^{l}}/\mathbb{F}_{q}}(\zeta_{l-m+j}c_{i})(p_{j}(\alpha_{i}))^{(l-m)(j+1)}
	\nonumber\\
	&&+\sum_{u=0}^{l-m-1}tr_{\mathbb{F}_{q^{l}}/\mathbb{F}_{q}}(\zeta_{u}c_{i})(p_{j}(\alpha_{i}))^{u(j+1)}.\nonumber\\
	%\label{eq:sizecorr1}
\end{IEEEeqnarray}

Substituing $c_{i}$ by $h(\alpha_{i})$ for all $i=1,\ldots,n$, we see that $(d_{1}^{j},\ldots,d_{n}^{j})=(T_{j}(h)(\alpha_{1}),\ldots,T_{j}(h)(\alpha_{n}))$ is the $j$-th row of the virtual projection code $\mathcal{C}_{P_{\alpha}}$ of $\mathcal{C}$. Now by the fact that $|A_{j}|=k/m$ for all $j$ and by the Corollary \ref{corola1}  we  know that $\tau_{P_{\alpha}}$ is given by
\begin{equation}\label{radiusfrac}
\tau_{P_{\alpha}}=\frac{1}{m+1}\left(mn+k\binom{m}{2}-\frac{k}{\alpha}\binom{m+1}{2}\right).
\end{equation}
As $\sum_{j=0}^{m-1}|A_{j}|=k$, using the Algorithm 1 it is possible to recover the codeword $c\in\mathcal{C}$ with failure probability given by Theorem \ref{falhap}  if $c$ has no more than $t\leqslant\tau_{P_{\alpha}}$ erros. 

Note that if $m=1$ then $\alpha=1/l$ and
\begin{eqnarray*}
	\tau_{P_{\alpha}}&=&\frac{1}{2}\left(n+k\binom{1}{2}-lk\binom{2}{2}\right)\\
%	&=&\frac{1}{2}\left(n-\frac{k}{1/l}\right)\\
	&=&\frac{1}{2}\left(n-\frac{k}{\alpha}\right)=\tau_{\alpha}.
\end{eqnarray*}

For $m\geqslant 2$, we would like to improve the fractional decoding radius of $\mathcal{C}$, it means that we are interested in the case $\tau_{P_{\alpha}}\geqslant\tau_{\alpha}$
\begin{equation}\label{cond}
\tau_{P_{\alpha}}=\frac{1}{m+1}\left(mn+k\binom{m}{2}-\frac{k}{\alpha}\binom{m+1}{2}\right)\geqslant\frac{n-k/\alpha}{2}.
\end{equation}
and it is possible to check that (\ref{cond}) is true if and only if

\begin{equation}\label{raiorest}
R=\frac{k}{n}\leqslant\frac{\alpha}{m(1-\alpha)+1}=\frac{m}{m(l-m)+l}.
\end{equation}

This can be summarized in the following theorem.

\begin{teo}
	Let $\mathcal{C}= RS(q^{l},n,k)$ be a Reed-Solomon Code with evaluation set $\mathcal{L}=\lbrace\alpha_{1},\ldots,\alpha_{n}\rbrace\subseteq\mathbb{F}_{q}$ and $\alpha=m/l$. If $m\geqslant 2$ and the rate of $\mathcal{C}$ is restricted as in (\ref{raiorest}) then the maximum $\alpha$-decoding radius of $\mathcal{C}$ using Algorithm 1 is
	
	\begin{equation}
	\tau_{P_{\alpha}}=\frac{1}{m+1}\left(mn+k\binom{m}{2}-\frac{k}{\alpha}\binom{m+1}{2}\right).
	\end{equation}
	Moreover, in this case $\tau_{P_{\alpha}}\geqslant\tau_{\alpha}$.
\end{teo}

\section{Failure Probability of the Algorithm I}
The failure probability  can be calculated in the same way that \cite{Schmidt} and \cite{Zeh}. %that is calculating the probability that the matrix

Note that the values of $T_{j_{1}}(e)(\alpha_{i})$ and $T_{j_{2}}(e)(\alpha_{i})$ do not depend of each other for all $j_{1},j_{2}\in\lbrace 0,\ldots,m-1\rbrace$ and we can assume that if $Y$ in (\ref{Y}) is corrupted by $t$ errors, that is, $Y=C+E$ where $E$ has $t$ non-zero columns, then each non-zero column is an independent random vector uniformly distributed over $\mathbb{F}_{q}^{m}\setminus\lbrace 0\rbrace$. Hence, we can apply Lemma 6  and Theorem 6 of  \cite{Schmidt} to upper bounded the failure probability of Algorithm 1.

%then we can assume that the errors columns  are uniformly distributed over all-nozero vectors of lenght $m$. we can apply Lemma 6  and Theorem 6 of  \cite{Schmidt}. %and $P_{r}$ is bounded by
%\begin{equation}\label{PR}
%P_{r}=\frac{q^{mr}}{(q^{m}-1)^{r}}q^{-\sum_{j=0}^{m-1}n-k_{j}-t}.
%\end{equation}
 %The upper bound on the failure probability can be calculated as in \cite{Zeh} Theorem 2. 
 
% \begin{defn}
%Two vectors non-zero vectors $v,v^{\prime}\in\mathbb{F}_{q}^{n}$ are said to be \textit{equivalent vectors} if $v^{\prime}=\beta v$ for an $\beta\in\mathbb{F}_{q}$.
 %\end{defn}
 
\begin{teo}\label{falhap}
Let $\mathcal{C}= RS(q^{l},n,k)$ be a Reed-Solomon Code with evaluation set $\mathcal{L}=\lbrace\alpha_{1},\ldots,\alpha_{n}\rbrace\subseteq\mathbb{F}_{q}$ and $\alpha=m/l$. If $m\geqslant 2$ and the rate of $\mathcal{C}$ is restricted as in (\ref{raiorest}). The probability for a decoding failure using the Algorithm 1 is upper bounded by 
\[P_{f_{\alpha}}(t)\leqslant\left(\frac{q^{m}-\frac{1}{q}}{q^{m}-1}\right)^{t}\frac{q^{-(m+1)(\tau_{P_{\alpha}}-t)}}{q-1}.\]
\end{teo}
%\begin{proof}
%We upper bound the failure probability if $t$ errors occurred using the number $M_{r}$ of non-equivalent vectors of length $t$ and certain weight $r$. $M_{r}$ can be calculated as in \cite{Schmidt}, Eq (19):
%\begin{equation*}
%M_{r}=\binom{t}{r}(q-1)^{r-1}
%\end{equation*}
%The probability $P_{r}$ is bounded by (\ref{PR}). Then, the failure probability is bounded by
%\begin{eqnarray*}
%P_{f}(t)&\leqslant&\sum_{r=1}^{t}P_{r}M_{r}\\
%&\leqslant&\frac{1}{q-1}\sum_{r=1}^{t}\binom{t}{r}(q-1)^{r}\frac{q^{mr}}{(q^{m}-1)^{r}}q^{-\sum_{j=0}^{m-1}n-k_{j}-t}\\
%&\leqslant&\frac{q^{-\sum_{j=0}^{m-1}n-k_{j}-t}   }{q-1}\sum_{r=0}^{t}\binom{t}{r}\left(\frac{(q-1)q^{m}}{q^{m}-1}\right)^{r}\\
%&=&\frac{q^{-\sum_{j=0}^{m-1}n-k_{j}-t}}{q-1}\left(\frac{q^{m+1}-1}{q^{m}-1}\right)^{t}\\
%&=&\frac{q^{-(m+1)\tau_{P_{\alpha}}+mt}}{q-1}\left(\frac{q^{m+1}-1}{q^{m}-1}\right)^{t}\\
%&=&\frac{q^{-(m+1)(\tau_{P_{\alpha}}-t)}}{q-1}\left(\frac{q^{m}-\frac{1}{q}}{q^{m}-1}\right)^{t}
%\end{eqnarray*}
%\end{proof}

\begin{Exemp}
	Let $\mathcal{C}= RS(2^{5},31,4)$ be a Reed-Solomon code with evaluation set $\mathcal{L}\in\mathbb{F}_{q}$ in this case the decoding radius of $\mathcal{C}$ is $\tau=13$ and $R\simeq 0.1290$. By definition $\alpha=\frac{m}{5}$ and $\frac{4}{31}\leqslant\frac{m}{5}<1$ thus $m\in\left\lbrace 1,2,3,4\right\rbrace$. Let $\alpha_{i}=\frac{i}{5}\ \text{for }\ i=2,3,4$ then for each $\alpha_{i}$ we have
	
	\begin{itemize}
		\item[a)] $\tau_{\alpha_{1}}=\tau_{P_{\alpha_{1}}}=5$.
		\item[b)] $\tau_{\alpha_{2}}=10<12=\tau_{P_{\alpha_{2}}}$.
		\item[c)] $\tau_{\alpha_{3}}=12<16=\tau_{P_{\alpha_{3}}}$.
		\item[d)] $\tau_{\alpha_{4}}=13<19=\tau_{P_{\alpha_{4}}}$.
	\end{itemize}
The failure probability of $c)$ is given in Table I.

\begin{table}[h]
	%	\centering
	\caption{FAILURE PROBABILITY $P_{f_{\alpha_{3}}}(t)$ FOR THE REED-SOLOMON CODE $ RS(2^{5},31,4)$.}
	\begin{tabular}{|c|cccc|}
		\hline	
		t & 12 & 13 & 14 & 15          \\ \hline
		$P_{f_{\alpha_{3}}}(t)$  & $2\times 10^{-6}$ & $7\times 10^{-5}$ & $2\times 10^{-3}$ & $8\times 10^{-2}$ \\ \hline
	\end{tabular}
\end{table}

\end{Exemp}

\begin{Exemp}
	Let $\mathcal{C}= RS(2^{5},31,6)$ be a Reed-Solomon code with evaluation set $\mathcal{L}\in\mathbb{F}_{q}$ in this case the decoding radius of $\mathcal{C}$ is $\tau=\lfloor\frac{n-k}{2}\rfloor=12$ and $R=\frac{k}{n}\simeq 0.1935$. By definition $\alpha=\frac{m}{5}$ and $\frac{6}{31}\leqslant\frac{m}{5}<1$ thus $m\in\left\lbrace 1,2,3,4\right\rbrace$. If we denoted $\alpha_{i}=\frac{i}{5}\ \text{for }\ i=2,3,4$ then for each $\alpha_{i}$ we have
	\begin{itemize}
		\item[a)] $\tau_{\alpha_{2}}=8>\tau_{P_{\alpha_{2}}}=7$. This is due to the fact that $R\simeq 0.1935$ and $\frac{\alpha_{2}}{2(1-\alpha_{2})+1}\simeq 0.1818$ that is (\ref{raiorest}) is not true in this case.
		
		\item[b)] $\tau_{\alpha_{3}}=10<12=\tau_{P_{\alpha_{3}}}$. 
		
		\item[c)] $\tau_{P_{\alpha_{4}}}=16\geq\tau_{\alpha_{4}}=11$. Note that $\tau_{P_{\alpha_{4}}}$ is even greater than the decoding radius of $\mathcal{C}$. So, without accessing the entire codeword it is possible to recover more than $\left\lfloor\frac{n-k}{2}\right\rfloor$ errors with failure probability given in the Table II.

\begin{table}[h]
	%	\centering
	\caption{FAILURE PROBABILITY $P_{f_{\alpha_{4}}}(t)$ FOR THE REED-SOLOMON CODE $ RS(2^{5},31,6)$.}
	\begin{tabular}{|c|cccccc|}
		\hline	
		t & 11 & 12 & 13 & 14 & 15 & 16           \\ \hline
		$P_{f_{\alpha_{4}}}(t)$  & $10^{-9}$ & $4\times 10^{-8}$ & $10^{-6}$ & $4\times 10^{-5}$ & $10^{-3}$ & $5\times 10^{-2}$  \\ \hline
	\end{tabular}
\end{table}

	\end{itemize}
% and $\tau_{\alpha_{4}}=11<17=\tau_{P_{\alpha_{4}}}$.
%Note that, $\tau_{\alpha_{2}}=8>7=\tau_{P_{\alpha_{2}}}$ 
%Note also that, 

\end{Exemp}

%\begin{Exemp}
%	Let $\mathcal{C}= RS(2^{5},31,11)$ be a Reed-Solomon code with evaluation set $\mathcal{L}\in\mathbb{F}_{q}$ in this case $R=\frac{11}{31}\simeq 0.3548$ and (\ref{raiorest}) is true if, and only if, $\alpha=\alpha_{4}=\frac{4}{5}$ in this case $\tau_{P_{\alpha_{4}}}=\tau=10>8=\tau_{\alpha_{4}}$.
%\end{Exemp}

\section*{Acknowledgment}

I would like to thank professor Alexander Barg for his help during my intership at University of Maryland and his
comments on this manuscript.

%\ifCLASSOPTIONcaptionsoff
% \newpage
%\fi


%\end{IEEEbiography}


%\newpage



%\vfill


\end{document}